\documentclass[11pt]{article}
\usepackage[T1]{fontenc}
\usepackage{amsmath,amsxtra,amssymb,latexsym,amscd,amsthm,pb-diagram,fancyhdr,euscript}
\usepackage{amsthm}
\usepackage{indentfirst}
\usepackage{epsfig}
\usepackage{mathrsfs}
\usepackage{slashed}
\usepackage{cases}
\usepackage{hyperref}
\usepackage{float}
\hypersetup{
   colorlinks,
    citecolor=blue,
    filecolor=black,
    linkcolor=blue,
    urlcolor=magenta
}
\hypersetup{linktocpage}

\renewcommand{\d}{\mathrm{d}}

\newcommand{\scri}{{\mathscr I}}

\newcommand{\hook}{{\setlength{\unitlength}{11pt}   
                   \begin{picture}(.833,.8)
                   \put(.15,.08){\line(1,0){.35}}
                   \put(.5,.08){\line(0,1){.5}}
                   \end{picture}}}
\newtheorem{definition}{Definition}
\newtheorem{theorem}{Theorem}
\newtheorem{proposition}{Proposition}

\newtheorem{lemma}{Lemma}
\newtheorem{remark}{Remark}

\newtheorem{lem}{Lemma}[section]

\begin{document}
\mbox{} \thispagestyle{empty}

\begin{center}
\bf{\Huge Peeling for tensorial wave equations on Schwarzschild spacetime} \\

\vspace{0.1in}

{PHAM Truong Xuan\footnote{Faculty of Pedagogy, VNU University of Education, Vietnam National University, Hanoi, 144 Xuan Thuy, Cau Giay, Hanoi, Viet Nam \\Email~: phamtruongxuan.k5@gmail.com}}
\end{center}

{\bf Abstract.} In this paper, we establish the asymptotic behaviour along outgoing and incoming radial geodesics, i.e., the peeling property for the tensorial Fackrell-Ipser and spin $\pm 1$ Teukolsky equations on Schwarzschild spacetime. Our method combines a conformal compactification with vector field techniques to prove the two-side estimates of the energies of tensorial fields through the future and past null infinity $\scri^\pm$ and the initial Cauchy hypersurface $\Sigma_0 = \left\{ t=0 \right\}$ in a neighbourhood of spacelike infinity $i_0$ far away from the horizon and future timelike infinity. Our results obtain the optimal initial data which guarantees the peeling at all orders. 



\tableofcontents

\section{Introduction}
The Peeling is a type of asymptotic behaviour of zero rest-mass fields initially discovered by R. Sachs  \cite{Sa61,Sa62}. Its initial formulation involved an expansion of the field in powers of $1/r$ along a null geodesic going out to infinity, and the alignment of a certain number of principal null directions of each term in the expansion along the null geodesic considered. Penrose introduced the conformal technique in the early 1960s \cite{Pe63,Pe64} and used it to establish that the peeling property is equivalent to the mere continuity of the rescaled field at null infinity \cite{Pe65}. 

The peeling for linearized gravity and for full gravity has been studied intensively (see Friedrich \cite{HFri2004}, Christodoulou-Klainerman \cite{ChriKla}, Corvino \cite{Co2000}, Chrusciel and Delay \cite{ChruDe2002,ChruDe2003}, Corvino-Schoen \cite{CoScho2003} and Klainerman-Nicol{\`o} \cite{KlaNi,KlaNi2002,KlaNi2003}) and is now fairly well understood, at least in the flat case. However, it is not yet clear, given an asymptotically flat spacetime, which class of initial data yields solutions that admit the peeling property at a given order, and whether these classes are smaller than in Minkowski spacetime or not.

The works of Mason and Nicolas \cite{MaNi2009,MaNi2012} were precisely aimed at answering this last question for the Schwarzschild metric, for scalar, Dirac and Maxwell fields. Their method combines the Penrose compactification of the spacetime and geometric energy estimates. By working in a neighbourhood of spacelike infinity on the compactified spacetime, one obtains energy estimates at all orders for the rescaled field, which control weighted Sobolev norms on $\scri$ in terms of similar norms on a Cauchy hypersurface and vice versa. The finiteness of the norms up to order $k$ (where $k\in \mathbb{N}$) at $\scri$ defines the peeling of order $k$. By the completion of smooth compactly supported data on the Cauchy hypersurface in the corresponding norms, one obtains the optimal classes of data ensuring that the associated solution peels at order $k$. The result does not strictly refer to the regularity near spacelike infinity $i_0$. Indeed, if the regularity is controlled in a neighbourhood of $i_0$ and on the full initial data hypersurface, it can be extended to the whole of $\scri$ by standard results for hyperbolic equations. The works \cite{MaNi2009, MaNi2012} put further a program of peeling on the asymptotic flat spacetimes. Continuing this program Nicolas and Pham \cite{NiXuan2019, Pha2019} established explicitly the peeling for linear (or semilinear) scalar wave and Dirac equations on Kerr spacetime which is not static and not symmetric spherical such as Schwarzschild spacetime.  

On the other hand, various dynamical constructions of spacetimes appear in physical reality and violate Sachs peeling property of gravitational radiation. These spacetimes do not possess a smooth null infinity $\scri$. In particular, Christodoulou \cite{Chri2002} assumed that the Bondi mass along $\scri^+$ decays with the rate predicted by the quadrupole approximation for a system of $N$
infalling masses coming from past infinity, combined with the assumption that there be no incoming radiation from past null infinity $\scri^-$, then the evolutions of Einstein-Vacuum equations (see \cite{ChriKla} for the detailed expression of these evolutions) does not admit a smooth conformal compactification, i.e., Penrose’s proposal of smooth conformal compactification of spacetime (or
smooth null infinity) fails for this case. Recently, motivated by \cite{Chri2002}, Kehrberger \cite{Ker2021} has constructed counter-examples to smooth null infinity by considering the solution of spherically symmetric Einstein-Scalar field system with positive Hawking mass and the polynomially decaying data on a timelike boundary (or on an ingoing null hypersurface) and no incoming radiation from past null infinity $\scri^-$ (see theorems 2.1, 2.2 and 2.4 in \cite{Ker2021}). Moreover, Kehrberger \cite{Ker2022} has also shown the failure of Sachs peeling in a neighbourhood of $i_0$ by proving the logarithmically modified Price’s law asymptotics near $\scri^+$ for the linear scalar wave equation $\Box_g\phi=0$ on Schwarzschild spacetime. In particular, he has considered conformally smooth initial data on an ingoing null hypersurface and vanishing data on $\scri^-$ for $\Box_g\phi=0$, then he has obtained precise asymptotics of the solution $\phi$ with logarithmic terms (see theorems 6.1 and 7.1 in \cite{Ker2022}). This also shows that the non-smoothness of null infinity near $i_0$ propagates and translates into logarithmic tails near $i^+$ (see also \cite{Ker22} for the early-time asymptotics (logarithmically modified Price’s law) for fixed-frequency solutions $\phi_l$ to the wave equation $\Box_g\phi_l=0$ with no incoming radiation condition on $\scri^-$ and polynomially decaying data $r\phi_l\simeq t^{-1}$ as $t\to -\infty$).  

Comeback to our work, we will explore the method in \cite{MaNi2009,NiXuan2019,Pha2019} to establish the peeling for tensorial Fackerell-Ipser and spin $\pm 1$ Teukolsky equations on Schwarzschild spacetime. 
First, we recall briefly some results about these equations: the spin $\pm 1$ Teukolsky equations are satisfied by the extreme components of Maxwell equation (see \cite{Ba1973,Pa2019}) and the spin $\pm 2$ Teukolsky equations arise from the linear and nonlinear stability problems of black hole spacetimes. We refer the reader to \cite{ABlu2,Da2019, El2020'} for the linear stability and \cite{Da2021, El2021, Kl2018, Kl2021, Kl2022} for the nonlinear stability. Moreover, the spin $\pm 1$ Teukolsky equations have been studied in some recent works by Giorgi \cite{El2019}, Ma \cite{Ma20} and Pasqualotto \cite{Pa2019}. On the other hand, the spin $\pm 2$ Teukolsky equations have been studied by Dafermos et al. \cite{Da2020} and Ma \cite{Ma202020}. In particular, the authors in \cite{Da2020,El2019,Ma20,Ma202020,Pa2019} have used $r^p$-method (see \cite{DaRo2010}) to establish the boundedness of energy and study time decays of the associated solutions of Teukolsky equations on Schwarzschild, Reissner-Nordstr\"om and Kerr spacetimes. The Price's law decays for Teukolsky equations on Kerr spacetime has been studied by Ma and Zhang \cite{Ma2021}, where they have established the global sharp decay for the spin $\pm s$ (where $s=0,1,2$) components, which are solutions to Teukolsky equations, in the black hole exterior and on the event horizon of a slowly rotating Kerr spacetime. On the other hand, the scattering theories for spin $\pm 1$ and $\pm 2$ Teukolsky equations on Schwarzschild spacetime have been studied by Pham \cite{Pham2020} and Masaood \cite{Masao}, respectively.

The tensorial Fackerell-Ipser equations are obtained by commuting the spin $\pm 1$ Teukolsky equations with the projected covariant derivatives $\slashed{\nabla}_L$ and $\slashed{\nabla}_{\underline L}$ on the $2$-sphere $\mathbb{S}^2_{(t,r)}$ at $(t,r)$, where $L$ and $\underline{L}$ are outgoing and incoming principal null directions respectively (see Subsection \ref{Equation}). Both the tensorial Fackerell-Ipser and spin $\pm 1$ Teukolsky equations are expressed under tensorial forms.
We consider these tensorial equations on a neighbourhood of spacelike infinity $\Omega_{u_0}^+$ which is foliated by a family of spacelike hypersurfaces $\left\{ \mathcal{H}_s \right\}_{0\leq s \leq 1}$. Using the stress-energy tensor for the tensorial linear Klein-Gordon equation we obtain the approximate conservation laws for the Fackerell-Ipser and Teukolsky equations and calculate energy fluxes of the associated solutions through $\mathcal{H}_s$ (see Lemma \ref{Energyfluxes}). In order to define the peeling we establish the two-side estimates of the energies through the future null infinity $\scri^+$ and the initial Cauchy hypersurface $\Sigma_0$ inside $\Omega_{u_0}^+$ (see theorems \ref{P01} and \ref{P011}). These estimates are obtained at all order $k$ (where $k\in \mathbb{N}$) of the projected covariant derivatives (see theorems \ref{PeelingFac} and \ref{PeelingTeu}). As a consequence, we can define the peeling at all order $k$ as well as the finiteness of energy norms of the solutions through the spacelike infinity $\scri^+$ inside $\Omega_{u_0}^+$. Moreover, we can also give the optimal initial data endowed Sobolev norm on $\Sigma_0$ which guarantees this definition. 

The paper is organized as follows: Section \ref{GeoEq} presents Schwarzschild spacetime and its Penrose's conformal compactification, the neighbourhood of the spacelike infinity $\Omega_{u_0}^+$ and its foliation, the Maxwell, tensorial Fackerell-Ipser and spin $\pm 1$ Teukolsky equations; Section \ref{Basic} contains the approximate conservation laws and energy fluxes of the associated solutions; Section \ref{Pee} relates the peeling for equations.

{\bf Notation.}\\
$\bullet$ We denote the bundle tangent to each $2$-sphere $\mathbb{S}^2(t,r)$ at $(t,r)$ by $\mathcal{B}$ and the vector space of all smooth sections of $\mathcal{B}$ by $\Gamma(\mathcal{B})$. The space of all $1$-forms on $\mathbb{S}^2_{(t,r)}$ is denoted by $\Lambda^1(\mathcal{B})$.\\ 
$\bullet$ We denote the orthogonal projection of covariant derivative $\nabla$ on $\mathcal{B}$ by $\slashed{\nabla}$.\\
$\bullet$ The space of $1$-forms on the unit $2$-sphere is denoted by $\Lambda^1(\mathbb{S}^2)$. The basic frame of $\Lambda^1(\mathbb{S}^2)$ is denoted by $(\slashed{\nabla}_{\partial_\theta},\slashed{\nabla}_{\partial_\varphi})$.\\
$\bullet$ We denote the space of smooth compactly supported scalar functions on $\mathcal{M}$ (a smooth manifold without boundary) by $\mathcal{C}_0^\infty(\mathcal{M})$. The space of smooth compactly supported $1$-forms in $\Lambda^1(\mathbb{S}^2)$ on $\mathcal{M}$ is denoted by $\mathcal{C}_0^\infty(\Lambda^1(\mathbb{S}^2)|_{\mathcal{M}})$.\\
$\bullet$ Let $f(x)$ and $g(x)$ be two real functions. We write $f \lesssim g$ if there exists a constant $C \in (0,+\infty)$ which is independent of the functions $f,\, g,$ and such that $f(x)\leq C g(x)$ for all parameter $x$. We write $f\simeq g$ if both $f\lesssim g$ and $g\lesssim f$ are valid.\\
{\bf Acknowledgements.} The author would like to thank the referee for his or her careful reading of the manuscript, give useful comments and related references which help us to improve this paper. This work is supported by the Vietnam Institute for Advanced Study in Mathematics (VIASM) 2023.

\section{Geometrical and analytical setting}\label{GeoEq} 

\subsection{Schwarzschild metric and Penrose's conformal compactification}
We consider the region outside the Schwarzschild black-hole $({\cal{M}}=\mathbb{R}_t\times ]2M,+\infty[_r\times \mathbb{S}^2,g)$, equipped with the Lorentzian metric $g$ given by
$$g = F \d t^2 - F^{-1}\d r^2 - r^2\d\mathbb{S}^2, \, F=F(r)=1-\mu,\, \mu=\frac{2M}{r},$$
where $\d \mathbb{S}^2$ is the euclidean metric on the unit $2$-sphere $\mathbb{S}^2$, and $M>0$ is the mass of the black hole.  

We recall that the Regge-Wheeler coordinate $r_*=r+2M \log(r-2M)$ satisfies $\d r=F\d r_*$. In the coordinates $(t,r_*,\theta,\varphi)$, the Schwarzschild metric takes the form
$$g = F(\d t^2- \d r_*^2) - r^2\d\mathbb{S}^2.$$
The retarded and advanced Eddington-Finkelstein coordinates $u$ and $v$ are defined by
$$u=t-r_*, \, v= t+r_*.$$
The outgoing and incoming principal null directions are 
$$L =\partial_v = \partial_t + \partial_{r_*}, \, \underline{L} = \partial_u = \partial_t - \partial_{r_*}$$
respectively.

Putting $\Omega = 1/r$ and $\hat{g} = \Omega^2g$. We obtain a conformal compactification of the exterior domain in the retarded variables $(u, \, R = 1/r, \, \theta,\, \varphi)$ that is $\left(\mathbb{R}_u\times \left[0,\dfrac{1}{2M}\right] \times \mathbb{S}^2, \hat{g} \right)$ with the rescaled metric 
\begin{equation*}
\hat{g} = R^2(1-2MR) \d u^2 - 2\d u\d R - \d\mathbb{S}^2.
\end{equation*}
The future null infinity $\scri^+$ and the past horizon $\mathfrak{H}^-$ are null hypersurfaces of the rescaled spacetime
$$\scri^+ = \mathbb{R}_u \times \left\{ 0\right\}_R \times \mathbb{S}^2, \, \mathfrak{H}^- = \mathbb{R}_u \times \left\{ 1/2M\right\}_R \times \mathbb{S}^2.$$
If we use the advanded variables $(v, \, R=1/r, \, \theta,\, \varphi)$, the rescaled metric $\hat{g}$ takes the form
\begin{equation*}
\hat{g} = R^2(1-2MR)\d v^2 + 2 \d v \d R - \d\mathbb{S}^2.
\end{equation*}
The past null infinity $\scri^-$ and the future horizon $\mathfrak{H}^+$ are described as the null hypersurfaces
$$\scri^- = \mathbb{R}_v \times \left\{ 0\right\}_R \times \mathbb{S}^2, \, \mathfrak{H}^+ = \mathbb{R}_v \times \left\{ 1/2M\right\}_R \times \mathbb{S}^2.$$
Penrose's conformal compactification of $\mathcal{M}$ is
$$\bar{\mathcal{M}} = \mathcal{M} \cup \scri^+ \cup \mathfrak{H}^+ \cup \scri^-\cup \mathfrak{H}^-\cup S_c^2,$$
where $S_c^2$ is the crossing sphere.

Note that the compactified spacetime $\bar{\mathcal{M}}$ is not compact. There are three ''points'' missing to the boundary: $i^+$, or future timelike infinity, defined as the limit point of uniformly timelike curves as $t\to + \infty$,
$i^-$, past timelike infinity, symmetric of $i^+$ in the distant past, and $i_0$, spacelike infinity, the limit point of uniformly spacelike curves as $r\to +\infty$. These ''points'', that can be described as $2$-spheres, are singularities of the rescaled metric $\hat{g}$.

In the retarded coordinates $(u,\, R, \, \theta,\, \varphi)$ we have the following relation
\begin{equation*}
\partial_R = - \frac{r^2}{F}(\partial_t + \partial_{r_*}) = -\frac{r^2}{F}L.
\end{equation*}
In the advanced coordinates $(v,\, R, \, \theta,\, \varphi)$ we have the following relation
\begin{equation*}
\partial_R = - \frac{r^2}{F}(\partial_t - \partial_{r_*}) = -\frac{r^2}{F}\underline{L}.
\end{equation*}

The scalar curvature of the rescaled metric $\hat g$ is
$$\mathrm{Scal}_{\hat g}= 12MR.$$

The volume form associated with the rescaled metric $\hat{g}$ are
$$\mathrm{dVol}_{\hat g} = -\d u \wedge \d R \wedge \d^2\omega = - \d v \wedge \d R \wedge \d^2 \omega,$$
where $\d^2\omega$ is the euclidean area element on unit $2$-sphere $\mathbb{S}^2$.

\subsection{Neighbourhood of spacelike infinity}\label{S22} 
By symmetry, we will study the peeling properties of Fackerell-Ipser and spin $\pm 1$ Teukolsky equations in domain $\left\{ t\geq 0 \right\}$ of $\bar{\mathcal{M}}$, the peeling in domain $\left\{ t \leq 0\right\}$ is done similarly. Following \cite{MaNi2009,MaNi2012}, we work on a future neighbourhood $\Omega_{u_0}^+ \, (\hbox{where   }u_0 \ll -1)$, of spacelike infinity $i_0$ that is sufficiently far away from the black hole and singularities. The neighbourhood $\Omega_{u_0}^+$ is bounded by a part of the Cauchy hypersurface $\Sigma_0 = \left\{ t = 0 \right\}$, a part of future null infinity $\scri^+$ and the following null hypersurface 
$$ \mathcal{S}_{u_0} = \left\{ u = u_0, \, t\geq 0 \right\} \text{ for } {u_0} \ll -1. $$
The neighbourhood $\Omega_{u_0}^+$ can be given precisely as 
$$ \Omega^+_{u_0} = {\cal I}^- ( \mathcal{S}_{u_0} ) \cap \{ t \geq 0\} $$
in the compactification domain $\bar{\mathcal{M}}$. We foliate $\Omega_{u_0}^+$ (the gray domain in Figure 1) by the following spacelike hypersurfaces
$${\cal H}_s \, = \, \left\{u = -sr_*; \, u \leq u_0 \right\}, \,\,\,\, 0 \leq s \leq 1, \text{ for a given } u_0\ll -1.$$
We explain the boundaries of $\Omega_{u_0}^+$ as follows (see Figure 1):
\begin{itemize}
\item If $s=1$, we have the first boundary ${\cal H}_1 = \left\{ t=0,\, u \leq u_0\right\}$, which is a part of $\Sigma_0$ inside $\Omega_{u_0}^+$.
\item If $s=0$, we have that the second boundary ${\cal H}_0$ is the limit of ${\cal H}_s$, when $s$ tends to zero. If $u$ is fixed and $s$ tends to $0$, then $r_*$ tends to $+\infty$ and $R$ tends to $0$. Therefore, we obtain that ${\cal H}_0  = \left\{ R=0,\, u \leq u_0\right\}$, which is a part of $\scri^+$ inside $\Omega_{u_0}^+$. We also denote ${\cal H}_0$ by $\scri^+_{u_0}$.
\item The third boundary is null hypersurface $ \mathcal{S}_{u_0} = \left\{ u = u_0, \, t\geq 0 \right\}$. Given $0 \leq s_1<s_2 \leq 1$, we denote by $\mathcal{S}_{u_0}^{s_1,s_2}$ the part of $\mathcal{S}_{u_0}$ between $\mathcal{H}_{s_1}$ and $\mathcal{H}_{s_2}$.
\end{itemize}
\begin{figure}[H]
\begin{center}
\includegraphics[scale=0.5]{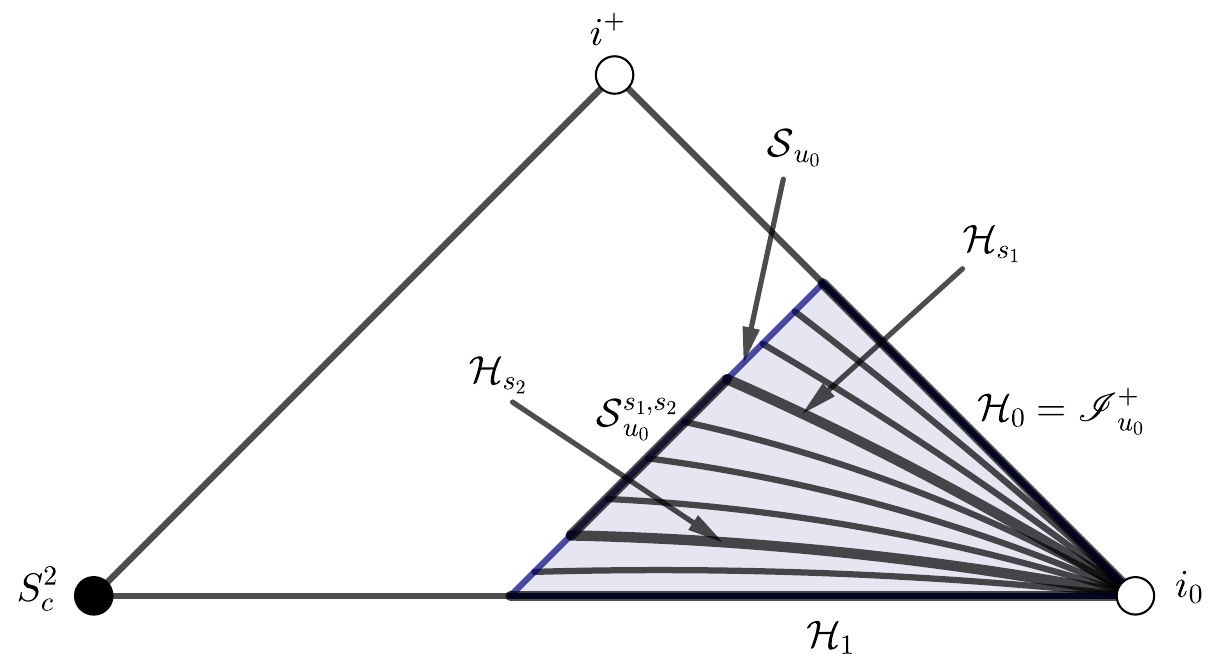}
\caption{Foliation $\left\{\mathcal{H}_s\right\}_{0\leq s\leq 1}$ of neighbourhood $\Omega_{u_0}^+$.}
\end{center}
\end{figure}
With the foliation $\left\{ {\cal H}_s \right\}_{0\leq s\leq 1}$, we choose an identifying vector field $\nu$ that satisfies $\nu(s)=1$ as follows
$$ \nu = r_*^2R^2 (1-2MR)|u|^{-1}\partial_R. $$
The $4-$volume measure $\mathrm{dVol}_{\hat g}$ can be decomposed  into the product of $\d s$ along the integral lines of $\nu^a$ and the $3$-volume measure
$$ \nu \hook \mathrm{dVol}_{\hat g} |_{{\cal H}_s}= -r_*^2 R^2(1-2MR)|u|^{-1} \d u \d^2\omega|_{{\cal H}_s}$$
on each slice ${\cal H}_s$.

We need the following lemma to establish simpler equivalent expressions in the next sections (see Lemma 2.1 in \cite{MaNi2012}):
\begin{lem}\label{epsilonestimates}
Let $\varepsilon >0$, then for $u_0\ll -1$, $|u_0|$ large enough, in $\Omega^+_{u_0}$, we have
\begin{equation*} 
r < r_* < (1+\varepsilon)r, \,  0 < |u|R < 1+ \varepsilon,\, 1 < r_*R < 1+\varepsilon \,\,\,  \mbox{and} \,\,\, 1-\varepsilon< 1- 2MR <1.
\end{equation*}
The factor $r_*^2 R^2(1-2MR)|u|^{-1}$ appearing in the $3$-volume measure $\nu \hook \mathrm{dVol} |_{{\cal H}_s}$ satisfies that
$$\frac{1+\varepsilon}{|u|} < r_*^2 R^2(1-2MR)|u|^{-1} < \frac{(1+\varepsilon)^2}{|u|}.$$
\end{lem}

\subsection{Maxwell, spin $\pm 1$ Teukolsky and tensorial Fackerell-Ipser equations}\label{Equation}
Let $\mathcal{F}$ be an antisymmetric $2$-form on the exterior domain of Schwarzschild black hole $\mathcal{M}$. The Maxwell equations take the form
\begin{equation*}
\d \mathcal{F} = 0, \,\,\, \d *\mathcal{F} = 0,
\end{equation*}
where $*$ denotes the Hodge dual operator of $2$-form, i.e,
$$(*\mathcal{F})_{\mu\nu} = \frac{1}{2}e_{\mu\nu\gamma\delta}\mathcal{F}^{\gamma\delta}.$$
The system can be reformulated as follows
\begin{equation*}
\nabla_{[\mu}\mathcal{F}_{\kappa\lambda]}=0, \,\,\, \nabla^\mu \mathcal{F}_{\mu\nu}=0,
\end{equation*}
where the square brackets denote the antisymmetrization of indices.

The Maxwell field $\mathcal{F}$ can be decomposed into $1$-forms $\alpha_a, \, \underline{\alpha}_a \in \Lambda^1(\mathcal{B})$ and $\rho,\, \sigma \in C^{\infty}(\mathcal{M})$ which are defined as follows
\begin{gather*}
\alpha(V):= \mathcal{F}(V,L), \, \underline{\alpha}(V):= \mathcal{F}(V,\underline{L}) \hbox{  for all  } V\in \Gamma(\mathcal{B}),\\
\rho:= \frac{1}{2} \left( 1 - \frac{2M}{r} \right)^{-1} \mathcal{F}(\underline{L},L),\, \sigma:= \frac{1}{2}e^{cd} \mathcal{F}_{cd},
\end{gather*}
where $e_{cd}\in \Lambda^2(\mathcal{B})$ is the volume form of $2$-sphere $\mathbb{S}^2_{(t,r)}$ at $(t,r)$.

Let $\mathcal{F}$ be in $\Lambda^2(\mathcal{M})$ such that $\mathcal{F}$ satisfies the Maxwell equation on $\mathcal{M}$. Then, we have the following formulas (see \cite[Proposition 3.6]{Pa2019}):
\begin{equation*}
\frac{1}{r}\slashed{\nabla}_L(r\underline{\alpha}_a) = - (1-\mu)(\slashed{\nabla}_a\rho - e_{ab}\slashed{\nabla}^b\sigma) 
\end{equation*} 
and
\begin{equation*}
\frac{1}{r}\slashed{\nabla}_{\underline L}(r\alpha_a) = (1-\mu)(\slashed{\nabla}_a\rho + e_{ab}\slashed{\nabla}^b\sigma).
\end{equation*}
From this, we can define the $1$-forms in $\Lambda^1(\mathcal{B})$:
\begin{equation}\label{Tran}
\phi_a := \frac{r^2}{F}\slashed{\nabla}_{\underline L}(r\alpha_a), \,\,\, \underline{\phi}_a := \frac{r^2}{F}\slashed{\nabla}_{L}(r\underline{\alpha}_a).
\end{equation}
Moreover, the extreme components $\alpha_a$ and $\underline{\alpha}_a$ satisfy the spin $\pm 1$ Teukolsky equations respectively (see original proof in \cite{Ba1973} and recent \cite[Proposition 3.6]{Pa2019}):
\begin{equation}\label{Teu1}
\slashed{\nabla}_L\slashed{\nabla}_{\underline L}(r\alpha_a) + \frac{2}{r}\left( 1-\frac{3M}{r} \right)\slashed{\nabla}_{\underline L}(r\alpha_a) - F\slashed\Delta(r\alpha_a) + \frac{F}{r^2}r\alpha_a=0,
\end{equation}
\begin{equation}\label{Teu2}
\slashed{\nabla}_L\slashed{\nabla}_{\underline L}(r\underline{\alpha}_a) - \frac{2}{r}\left( 1-\frac{3M}{r} \right)\slashed{\nabla}_{L}(r\underline{\alpha}_a) - F\slashed\Delta(r\underline{\alpha}_a) + \frac{F}{r^2}r\underline{\alpha}_a=0,
\end{equation}
where $F=1-2MR$, $\slashed \nabla$ and $\slashed\Delta = \frac{1}{r^2}\Delta_{\mathbb{S}^2}$ are the orthogonal projection of covariant derivative $\nabla$ and covariant laplacian operator on the bundle tangent  $\mathcal{B}$ of $2$-sphere $\mathbb{S}^2_{(t,r)}$ respectively.

The tensorial Fackerell-Ipser equations are established from the spin $\pm 1$ Teukolsky equations by the following proposition (see \cite[Proposition 3.7]{Pa2019}):
\begin{proposition}\label{relationEq}
Suppose that $(\alpha_a,\underline{\alpha}_a,\rho,\sigma)$ satisfy the Maxwell equation, then the $1$-forms $\phi_a$ and $\underline{\phi}_a$ satisfy the following tensorial Fackerell-Ipser equations
\begin{equation}\label{Fac01}
\slashed{\Box}_g(\phi_a) + \frac{1}{r^2}\phi_a=0,
\end{equation}
\begin{equation}\label{Fac02}
\slashed{\Box}_g(\underline{\phi}_a) + \frac{1}{r^2}\underline{\phi}_a=0,
\end{equation}
where we denote the tensorial wave operator by
$$\slashed{\Box}_g = \frac{1}{F}\slashed{\nabla}_L\slashed{\nabla}_{\underline L} - \slashed{\Delta}.$$
\end{proposition}
\begin{remark}
The tensorial Fackerell-Ipser operator has the same form as the rescaled tensorial wave operator by multiplying the factor $r^2$ due to
\begin{subnumcases}{\slashed{\Box}_{\hat g} = \frac{r^2}{F}\slashed{\nabla}_L\slashed{\nabla}_{\underline{L}} - \slashed{\Delta}_{\mathbb{S}^2} =}
-2\slashed{\nabla}_v\slashed{\nabla}_R - \slashed{\nabla}_R R^2(1-2MR)\slashed{\nabla}_R - \slashed{\Delta}_{\mathbb{S}^2} & $\hbox{in  } (v,\, R,\, \theta,\, \varphi),$ \cr
-2\slashed{\nabla}_u\slashed{\nabla}_R - \slashed{\nabla}_R R^2(1-2MR)\slashed{\nabla}_R - \slashed{\Delta}_{\mathbb{S}^2} & $\hbox{in  } (u,\, R,\, \theta,\, \varphi).$ \nonumber  
\end{subnumcases}
\end{remark}

In the advanced coordinates $(v,\,R,\,\theta,\, \varphi)$ the tensorial Fackerell-Ipser and spin $+1$ Teukolsky equations \eqref{Fac01} and \eqref{Teu1} have the following forms
\begin{equation}\label{ReFac01}
-2\slashed{\nabla}_v\slashed{\nabla}_R\phi_a - \slashed{\nabla}_R R^2(1-2MR)\slashed{\nabla}_R\phi_a - \slashed{\Delta}_{\mathbb{S}^2}\phi_a + \phi_a = 0
\end{equation}
and
\begin{equation}\label{ReTeu1}
-2\slashed{\nabla}_v\slashed{\nabla}_R\widehat{\alpha}_a - \slashed{\nabla}_R R^2(1-2MR)\slashed{\nabla}_R\widehat{\alpha}_a - \slashed{\Delta}_{\mathbb{S}^2}\widehat{\alpha}_a - 2R(1-3MR)\slashed{\nabla}_R\widehat{\alpha}_a + \widehat{\alpha}_a = 0, \hbox{   } \widehat{\alpha}_a = r\alpha_a
\end{equation}
respectively.

In the retarded coordinates $(u,\,R,\,\theta,\,\varphi)$ the tensorial Fackerell-Ipser and spin $-1$ Teukolsky equations \eqref{Fac02} and \eqref{Teu2} have the following forms
\begin{equation}\label{ReFac02}
-2\slashed{\nabla}_u\slashed{\nabla}_R\underline{\phi}_a - \slashed{\nabla}_R R^2(1-2MR)\slashed{\nabla}_R\underline{\phi}_a - \slashed{\Delta}_{\mathbb{S}^2}\underline{\phi}_a + \underline{\phi}_a = 0
\end{equation}
and
\begin{equation}\label{ReTeu2}
-2\slashed{\nabla}_u\slashed{\nabla}_R\widehat{\underline{\alpha}}_a - \slashed{\nabla}_R R^2(1-2MR)\slashed{\nabla}_R\widehat{\underline{\alpha}}_a - \slashed{\Delta}_{\mathbb{S}^2}\widehat{\underline{\alpha}}_a + 2R(1-3MR)\slashed{\nabla}_R\widehat{\underline{\alpha}}_a + \widehat{\underline{\alpha}}_a = 0, \hbox{  } \widehat{\underline\alpha}_a = r\underline{\alpha}_a
\end{equation}
respectively.

In the rest of this paper, we will establish the peeling for the tensorial Fackerell-Ipser and spin $-1$ Teukolsky equations \eqref{ReFac02} and \eqref{ReTeu2} in the neighbourhood $\Omega_{u_0}^+$, i.e., the asymptotic behaviours of the associated solutions along outgoing radial geodesics in coordinates $(u,\, R,\, \theta,\, \varphi)$. The constructions for the equations \eqref{ReFac01} and \eqref{ReTeu1} are done similarly in coordinates $(v,\, R,\, \theta,\, \varphi)$.

\section{Basic formulae}\label{Basic}
\subsection{Approximate conservation laws}
For a $1$-form $\xi_a \in \Lambda^1(\mathcal{B})$ on the $2$-sphere $\mathbb{S}^2_{(t,r)}$ we define
$$\xi^a = \xi_b g_{\mathbb{S}^2}^{ab}, \hbox{    } \slashed{\nabla}_{\underline L}\xi^a = \slashed{\nabla}_{\underline L}\xi_b g_{\mathbb{S}^2}^{ab} \hbox{  and  } \slashed{\nabla}_L \xi^a = \slashed{\nabla}_L \xi_b g_{\mathbb{S}^2}^{ab}.$$
Putting
$$|\xi_a|^2 := \xi_a \xi^a , \, |\slashed{\nabla}_u\xi_a|^2 := \slashed{\nabla}_u{\xi}_a \slashed{\nabla}_u{\xi}^a  \hbox{  and  } |\slashed{\nabla}_R\xi_a|^2 := \slashed{\nabla}_R\xi_a \slashed{\nabla}_R\xi^a,$$
where the scalar product depends on the metric $g_{\mathbb{S}^2}$.

We define the stress-energy tensor for the tensorial linear Klein-Gordon equation $\slashed{\Box}_{\hat g} \xi_a + \xi_a = 0$ as follows
\begin{equation}
\mathbb{T}_{cd}(\xi_a) = \mathbb{T}_{(cd)} (\xi_a) = \slashed{\nabla}_c \xi_a \slashed{\nabla}_d \xi^a - \frac{1}{2}\left< \slashed{\nabla}\xi_a, \slashed{\nabla}\xi^a\right >_{\hat g} \hat{g}_{cd} + \frac{1}{2}|\xi_a|^2\hat{g}_{cd}.
\end{equation}

In order to obtain the approximate conservation laws for \eqref{ReFac02} and \eqref{ReTeu2} we use the Morawetz vector field
$$T^c\partial_c = u^2\partial_u -2(1+uR)\partial_R.$$
By using Lie derivative of the rescaled Schwarzschild metric $\hat{g}$ follows the Morawetz vector field $T$, we can derive that (see \cite{MaNi2009,NiXuan2019}):
\begin{equation}
\slashed{\nabla}_cT_d = \slashed{\nabla}_{(c}T_{d)} = 4MR^2(3+uR)\d u^2.
\end{equation}

For a solution $\underline{\phi}_a$ of the tensorial Fackerell-Ipser equation \eqref{ReFac02}, we have
\begin{equation}\label{Non1}
T^d\slashed{\nabla}^c \mathbb{T}_{cd}(\underline{\phi}_a) = \left( \slashed{\Box}_{\hat g} \underline{\phi}_a + \underline{\phi}_a\right)\slashed{\nabla}_T\underline{\phi}^a=0.
\end{equation}
For a solution $\widehat{\underline\alpha}_a$ of the Teukolsky equation \eqref{ReTeu2}, we have
\begin{equation}\label{Non2}
T^d\slashed{\nabla}^c \mathbb{T}_{cd}(\widehat{\underline{\alpha}}_a) = \left( \slashed{\Box}_{\hat g} \widehat{\underline{\alpha}}_a + \widehat{\underline{\alpha}}_a\right)\slashed{\nabla}_T\widehat{\underline{\alpha}}^a = -2R(1-3MR)\slashed{\nabla}_R\widehat{\underline\alpha}_a\slashed{\nabla}_T \widehat{\underline{\alpha}}^a.
\end{equation}

Setting
\begin{eqnarray}
J_c(\underline{\phi}_a) := T^d \mathbb{T}_{cd}(\underline{\phi}_a) \hbox{  and  } J_c(\widehat{\underline{\alpha}}_a):= T^d \mathbb{T}_{cd}(\widehat{\underline{\alpha}}_a).
\end{eqnarray}
From \eqref{Non1} and \eqref{Non2} the nonlinear energy currents $J_c(\underline{\phi}_a)$ and $J_c(\widehat{\underline{\alpha}}_a)$ satisfy the following approximate conservation laws
\begin{equation}\label{zero1}
\slashed{\nabla}^cJ_c(\underline{\phi}_a) = (\slashed{\nabla}_cT_d) \mathbb{T}^{cd}(\underline{\phi}_a) = 4MR^2(3+uR)|\slashed{\nabla}_R\underline{\phi}_a|^2.
\end{equation}
and
\begin{eqnarray}\label{zero2}
\slashed{\nabla}^cJ_c(\widehat{\underline{\alpha}}_a) &=&  -2R(1-3MR)\slashed{\nabla}_R\widehat{\underline\alpha}_a\slashed{\nabla}_T \widehat{\underline{\alpha}}^a + (\slashed{\nabla}_cT_d) \mathbb{T}^{cd}(\widehat{\underline{\alpha}}_a)\cr
&=& -2R(1-3MR)\slashed{\nabla}_R\widehat{\underline\alpha}_a \left( u^2\slashed{\nabla}_u \widehat{\underline{\alpha}}^a -2(1+uR)\slashed{\nabla}_R\widehat{\underline{\alpha}}^a \right) \cr
&&+ 4MR^2(3+uR)|\slashed{\nabla}_R\widehat{\underline{\alpha}}_a|^2
\end{eqnarray}
respectively.

Because of the symmetries of Schwarzschild spacetime, we have for any $k\in \mathbb{N}$:
\begin{equation*}\label{Higher1}
\slashed{\Box}_{\hat g}\slashed{\nabla}^k_u\underline{\phi}_a + \slashed{\nabla}^k_u\underline{\phi}_a = 0,\, \slashed{\Box}_{\hat g}\slashed{\nabla}^k_u\widehat{\underline{\alpha}}_a + 2R(1-3MR)\slashed{\nabla}_R\slashed{\nabla}_u^k\widehat{\underline{\alpha}}_a + \slashed{\nabla}^k_u\widehat{\underline{\alpha}}_a = 0
\end{equation*}
and
\begin{equation*}\label{Higher2}
\slashed{\Box}_{\hat g}\slashed{\nabla}^k_{\mathbb{S}^2}\underline{\phi}_a + \slashed{\nabla}^k_{\mathbb{S}^2}\underline{\phi}_a = 0, \, \slashed{\Box}_{\hat g}\slashed{\nabla}^k_{\mathbb{S}^2}\widehat{\underline{\alpha}}_a + 2R(1-3MR)\slashed{\nabla}_R\slashed{\nabla}_{\mathbb{S}^2}^k\widehat{\underline{\alpha}}_a + \slashed{\nabla}^k_{\mathbb{S}^2}\widehat{\underline{\alpha}}_a = 0
\end{equation*}
Therefore, the approximate conservation laws \eqref{zero1} and \eqref{zero2} are valid at high order for $\slashed{\nabla}_u^k$ and $\slashed{\nabla}_{\mathbb{S}^2}^k$.

We commute the operator $\slashed{\nabla}_R$ into the equations \eqref{ReFac02} and \eqref{ReTeu2} to obtain
\begin{equation*}
\slashed{\Box}_{\hat g} \slashed{\nabla}_R\underline{\phi}_a + \slashed{\nabla}_R\underline{\phi}_a = 2(1-3M)R\slashed{\nabla}_R^2\underline{\phi}_a - 2(1-6MR)\slashed{\nabla}_R\underline{\phi}_a.
\end{equation*}
and
\begin{equation*}
\slashed{\Box}_{\hat g} \slashed{\nabla}_R \widehat{\underline{\alpha}}_a + \slashed{\nabla}_R\widehat{\underline{\alpha}}_a = 6MR(R-1)\slashed{\nabla}^2_R\widehat{\underline{\alpha}}_a - 4(1-6MR)\slashed{\nabla}_R\widehat{\underline{\alpha}}_a.
\end{equation*}
Therefore, we obtain the approximate conservation laws for $\slashed{\nabla}_R\underline{\phi}_a$ and $\slashed{\nabla}_R\widehat{\underline{\alpha}}_a$ as follows
\begin{eqnarray}\label{order1}
\slashed{\nabla}^cJ_c(\slashed{\nabla}_R\underline{\phi}_a) &=& \left( 2(1-3M)R\slashed{\nabla}_R^2\underline{\phi}_a - 2(1-6MR)\slashed{\nabla}_R\underline{\phi}_a \right)\slashed{\nabla}_T\slashed{\nabla}_R\underline{\phi}^a + (\slashed{\nabla}_cT_d) \mathbb{T}^{cd}(\slashed{\nabla}_R\underline{\phi}_a)\cr
&=& \left( 2(1-3M)R\slashed{\nabla}_R^2\underline{\phi}_a - 2(1-6MR)\slashed{\nabla}_R\underline{\phi}_a \right)\left( u^2\slashed{\nabla}_u\slashed{\nabla}_R\underline{\phi}^a - 2(1+uR)\slashed{\nabla}_R^2\underline{\phi}^a \right)\cr
&&+ 4MR^2(3+uR)|\slashed{\nabla}^2_R\underline{\phi}_a|^2.
\end{eqnarray}
and
\begin{eqnarray}\label{order2}
\slashed{\nabla}^cJ_c(\slashed{\nabla}_R\widehat{\underline{\alpha}}_a) &=& \left( 6MR(R-1)\slashed{\nabla}_R^2\widehat{\underline{\alpha}}_a - 4(1-6MR)\slashed{\nabla}_R\widehat{\underline{\alpha}}_a \right)\slashed{\nabla}_T\slashed{\nabla}_R\widehat{\underline{\alpha}}^a + (\slashed{\nabla}_cT_d) \mathbb{T}^{cd}(\slashed{\nabla}_R\widehat{\underline{\alpha}}_a)\cr
&=& \left( 6MR(R-1)\slashed{\nabla}_R^2\widehat{\underline{\alpha}}_a - 4(1-6MR)\slashed{\nabla}_R\widehat{\underline{\alpha}}_a \right)\left( u^2\slashed{\nabla}_u\slashed{\nabla}_R\widehat{\underline{\alpha}}^a - 2(1+uR)\slashed{\nabla}_R^2\widehat{\underline{\alpha}}^a \right)\cr
&&+ 4MR^2(3+uR)|\slashed{\nabla}^2_R\widehat{\underline{\alpha}}_a|^2.
\end{eqnarray}
By the same way we can obtain the approximate conservation laws for $\slashed{\nabla}_R^k\underline{\phi}_a$ and $\slashed{\nabla}_R^k\widehat{\underline{\alpha}}_a$ for all $k\in \mathbb{N}$.

\subsection{Energy fluxes}
Moreover, we follow the convention used by Penrose and Rindler \cite{PeRi84} about the Hodge dual of a $1$-form $\alpha$ on a spacetime $({\cal M},g)$ (i.e. a $4-$dimensional Lorentzian manifold that is oriented and time-oriented):
\begin{equation*}
(*\alpha)_{abcd} = e_{abcd}{\alpha}^d,
\end{equation*}
where $e_{abcd}$ is the volume form on $({\cal M},g)$, denoted simply by $\mathrm{dVol}_g$. We shall use the following differential operator of the Hodge star
\begin{equation*}
\d *\alpha = -\frac{1}{4}(\nabla_a\alpha^a)\mathrm{dVol}^4_g.
\end{equation*}
If $\mathcal{S}$ is the boundary of a bounded open set $\Omega$ and has outgoing
orientation, using Stokes theorem, we have
\begin{equation}\label{Stokesformula}
-4\int_{\mathcal{S}}*\alpha = \int_{\Omega}(\nabla_a\alpha^a)\mathrm{dVol}^4_g.
\end{equation}
Let $\underline{\phi}_a$ and $\widehat{\underline{\alpha}}_a$ be solutions of \eqref{ReFac02} and \eqref{ReTeu2} with smooth and compactly supported initial data on the rescaled spacetime $(\bar{\mathcal{M}},\hat{g})$. By using \eqref{Stokesformula} we define the rescaled energy fluxes associated with the Morawertz vector field $T$, through an oriented hypersurface $\mathcal{S}$ as follows
\begin{equation}\label{e1}
\mathcal{E}_{\mathcal {S}}(\underline{\phi}_a) = -4\int_{\mathcal{S}} *J_c(\underline{\phi}_a)\d x^c = \int_{\mathcal{S}} J_c(\underline{\phi}_a)\hat{N}^c\hat{L}\hook \mathrm{dVol}_{\hat g}
\end{equation}
and
\begin{equation}\label{e2}
\mathcal{E}_{\mathcal {S}}(\widehat{\underline{\alpha}}_a) = -4\int_{\mathcal{S}} *J_c(\widehat{\underline{\alpha}}_a)\d x^c = \int_{\mathcal{S}} J_c(\widehat{\underline{\alpha}}_a)\hat{N}^c\hat{L}\hook \mathrm{dVol}_{\hat g}
\end{equation}
where $\hat{L}$ is a transverse vector to $\mathcal{S}$ and $\hat{N}$ is the normal vector field to $\mathcal{S}$ such that
$\hat{L}^a\hat{N}_a=1$.

We recall the following Poincar\'e-type inequality (see \cite[Lemma 4.2]{MaNi2009}):
\begin{lemma}\label{PoinIne}
Given $u_0<0$, there exists a constant $C > 0$ such that for any $f\in C_0^\infty(\mathbb{R})$ such that
$$\int_{-\infty}^{u_0}(f(u))^2 \d u \leq C \int_{-\infty}^{u_0}|u|^2(f'(u))^2\d u.$$
\end{lemma}
Using Lemma \ref{PoinIne}, we can give the simpler equivalent expressions of energy fluxes for equations \eqref{ReFac02} and \eqref{ReTeu2} across the leaves of the foliation $\mathcal{H}_s$ of $\Omega_{u_0}^+$:
\begin{lemma}\label{Energyfluxes}
For $|u_0|$ large enough, the energy fluxes of $\underline{\phi}_a$ and $\widehat{\underline{\alpha}}_a $ through the hypersurface $\mathcal{H}_s, \, 0<s\leq 1$ and $H_0 = \scri_{u_0}^+$ have the following simpler equivalent expressions
\begin{equation*}
\mathcal{E}_{\mathcal{H}_s}(\underline{\phi}_a) \simeq \int_{\mathcal{H}_s} \left( |u|^2|\slashed{\nabla}_u\underline{\phi}_a|^2 + \frac{R}{|u|}|\slashed{\nabla}_R\underline{\phi}_a|^2 + |\slashed{\nabla}_{\mathbb{S}^2}\underline{\phi}_a|^2 + |\underline{\phi}_a|^2 \right)\d u \d^2\omega,
\end{equation*}
\begin{equation*}
\mathcal{E}_{\mathcal{H}_s}(\widehat{\underline{\alpha}}_a) \simeq \int_{\mathcal{H}_s} \left( |u|^2|\slashed{\nabla}_u\widehat{\underline{\alpha}}_a|^2 + \frac{R}{|u|}|\slashed{\nabla}_R\widehat{\underline{\alpha}}_a|^2 + |\slashed{\nabla}_{\mathbb{S}^2}\widehat{\underline{\alpha}}_a|^2 + |\widehat{\underline{\alpha}}_a|^2 \right)\d u \d^2\omega
\end{equation*}
and
\begin{equation*}
\mathcal{E}_{\scri_{u_0}^+}(\underline{\phi}_a) \simeq \int_{\mathcal{H}_s} \left( |u|^2|\slashed{\nabla}_u\underline{\phi}_a|^2 + |\slashed{\nabla}_{\mathbb{S}^2}\underline{\phi}_a|^2 + |\underline{\phi}_a|^2 \right)\d u \d^2\omega,
\end{equation*}
\begin{equation*}
\mathcal{E}_{\scri_{u_0}^+}(\widehat{\underline{\alpha}}_a) \simeq \int_{\mathcal{H}_s} \left( |u|^2|\slashed{\nabla}_u\widehat{\underline{\alpha}}_a|^2 + |\slashed{\nabla}_{\mathbb{S}^2}\widehat{\underline{\alpha}}_a|^2 + |\widehat{\underline{\alpha}}_a|^2 \right)\d u \d^2\omega.
\end{equation*}
Here, we denote that 
$$|\slashed{\nabla}_{\mathbb{S}^2}\underline{\phi}_a|^2 = |\slashed{\nabla}_{\partial_\theta}\underline{\phi}_a|^2 + \frac{1}{\sin^2\theta}|\slashed{\nabla}_{\partial_\varphi}\underline{\phi}_a|^2, \,\, |\slashed{\nabla}_{\mathbb{S}^2}\widehat{\underline{\alpha}}_a|^2 = |\slashed{\nabla}_{\partial_\theta}\widehat{\underline{\alpha}}_a|^2 + \frac{1}{\sin^2\theta}|\slashed{\nabla}_{\partial_\varphi}\widehat{\underline{\alpha}}_a|^2.$$
\end{lemma}
\begin{proof}
The proof is similar to the case of scalar wave equation using formulas \eqref{e1}, \eqref{e2} and Lemma \ref{PoinIne} (see Lemma 4.1 in \cite{MaNi2009} or Lemma 4.1 in \cite{NiXuan2019}).
\end{proof}

\section{Peeling}\label{Pee}
\subsection{Peeling for tensorial Fackerell-Ipser equations}\label{Sub1}
Integrating the approximate conservation law \eqref{zero1} on the domain
\begin{equation}\label{Omega}
\Omega_{u_0}^{s_1,s_2}:= \Omega_{u_0}^+\cap \left\{ s_1 \leq s \leq s_2 \right\} \hbox{  with  } 0 \leq s_1<s_2\leq 1, 
\end{equation}
we get
\begin{eqnarray}\label{zero}
&&\left| \mathcal{E}_{\mathcal{H}_{s_1}}(\underline{\phi}_a) + \mathcal{E}_{\mathcal{S}^{s_1,s_2}_{u_0}}(\underline{\phi}_a) - \mathcal{E}_{\mathcal{H}_{s_2}}(\underline{\phi}_a) \right|\cr
&\simeq& \left| \int_{s_1}^{s_2}\int_{\mathcal{H}_s} 4MR^2(3+uR)|\slashed{\nabla}_R\underline{\phi}_a|^2 \frac{1}{|u|}\d u \d^2\omega \d s \right|\cr
&\lesssim& \int_{s_1}^{s_2}\int_{\mathcal{H}_s} 4MR^2(3+uR)|\slashed{\nabla}_R\underline{\phi}_a|^2 \frac{1}{|u|}\d u \d^2\omega \d s\cr
&\lesssim& \int_{s_1}^{s_2}\int_{\mathcal{H}_s} \frac{R^2}{|u|}|\slashed{\nabla}_R\underline{\phi}_a|^2\d u \d^2\omega \d s\cr
&\lesssim& \int_{s_1}^{s_2}\int_{\mathcal{H}_s} \frac{1}{\sqrt s}\frac{R}{|u|}|\slashed{\nabla}_R\underline{\phi}_a|^2\d u \d^2\omega \d s\cr
&\lesssim& \int_{s_1}^{s_2}\frac{1}{\sqrt s}\mathcal{E}_{\mathcal{H}_s}(\underline{\phi}_a) \d s,
\end{eqnarray}
Here, we used the fact that on $\mathcal{H}_s$, we have $u=-sr_*$, then  $\dfrac{R^2}{|u|} \simeq \dfrac{s}{|u|}\dfrac{R}{|u|}\leqslant \dfrac{1}{\sqrt s}\dfrac{R}{|u|}$  for $u\leq u_0 \ll -1$ and $R$ large enough. The inequality \eqref{zero} leads to
\begin{eqnarray}
&&\mathcal{E}_{\mathcal{H}_{s_1}}(\underline{\phi}_a) \lesssim \mathcal{E}_{\mathcal{H}_1}(\underline{\phi}_a) + \int_{s_1}^1 \frac{1}{\sqrt s}\mathcal{E}_{\mathcal{H}_s}(\underline{\phi}_a) \d s, \label{0}\\
&&\mathcal{E}_{\mathcal{H}_{s_2}}(\underline{\phi}_a) \lesssim \mathcal{E}_{\mathcal{S}^+_{u_0}}(\underline{\phi}_a) + \mathcal{E}_{\scri^+_{u_0}}(\underline{\phi}_a) + \int_0^{s_2}\frac{1}{\sqrt s}\mathcal{E}_{\mathcal{H}_s}(\underline{\phi}_a) \d s \label{1}
\end{eqnarray}
for all $0\leq s_1<s_2\leq 1$. Note that, the factor $\dfrac{1}{\sqrt s}$ is not really necessary here because we can estimate $\dfrac{R^2}{|u|}\leq \dfrac{R}{|u|}$ in inequality \eqref{zero}, and then inequalities $\eqref{0}$ and $\eqref{1}$ have not the factor $\dfrac{1}{\sqrt{s}}$. However, we still present $\dfrac{1}{\sqrt s}$  to serve the estimates at higher order below (see inequality \eqref{Gron}).

Since the function $\dfrac{1}{\sqrt{s}}$ is integrable on $[0,\,1]$, we use Gronwall’s inequality for inequalities \eqref{0} and \eqref{1} with the scalar function $y(s)=\mathcal{E}_{\mathcal{H}_s}(\underline{\phi}_a)$ and get the following result for energy estimates at zero order.
\begin{theorem}\label{P01}
For $u_0<0$ and $|u_0|$ large enough  and for any smooth compactly supported initial data
at $\Sigma_0= \left\{ t = 0 \right\}$, the associated solutions $\underline{\phi}_a$ of \eqref{ReFac02}
satisfying that
$$\mathcal{E}_{\scri^+_{u_0}}(\underline{\phi}_a) \lesssim \mathcal{E}_{\mathcal{H}_1}(\underline{\phi}_a),$$
$$\mathcal{E}_{\mathcal{H}_1}(\underline{\phi}_a) \lesssim \mathcal{E}_{\scri^+_{u_0}}(\underline{\phi}_a) + \mathcal{E}_{{\mathcal{S}^+_{u_0}}}(\underline{\phi}_a).$$
\end{theorem}
Since the approximate conservation law \eqref{zero1} is valid for $\slashed{\nabla}_u^k\underline{\phi}_a$ and $\slashed{\nabla}_{\mathbb{S}^2}^k\underline{\phi}_a$ with $k\in \mathbb{N}$, we have the following estimates at higher order for all projected covariant derivatives.
\begin{theorem}\label{P1}
For $u_0<0$ and $|u_0|$ large enough  and for any smooth compactly supported initial data
at $\Sigma_0= \left\{ t = 0 \right\}$, the associated solutions $\underline{\phi}_a$ of \eqref{ReFac02}
satisfying that
$$\mathcal{E}_{\scri^+_{u_0}}(\slashed{\nabla}_u^k\slashed{\nabla}_{\mathbb{S}^2}^l\underline{\phi}_a) \lesssim \mathcal{E}_{\mathcal{H}_1}(\slashed{\nabla}_u^k\slashed{\nabla}_{\mathbb{S}^2}^l\underline{\phi}_a),$$
$$\mathcal{E}_{\mathcal{H}_1}(\slashed{\nabla}_u^k\slashed{\nabla}_{\mathbb{S}^2}^l\underline{\phi}_a) \lesssim \mathcal{E}_{\scri^+_{u_0}}(\slashed{\nabla}_u^k\slashed{\nabla}_{\mathbb{S}^2}^l\underline{\phi}_a) + \mathcal{E}_{\mathcal{H}_{\mathcal{S}^+_{u_0}}}(\slashed{\nabla}_u^k\slashed{\nabla}_{\mathbb{S}^2}^l\underline{\phi}_a),$$
where $k,l\in \mathbb{N}$.
\end{theorem}

Since on the hypersurface $\mathcal{H}_s$, we have $s=-\dfrac{u}{r_*}\simeq -uR$, then we obtain the equivalence
\begin{equation}\label{BasicEq}
\frac{1}{|u|} \simeq \frac{1}{\sqrt s}\sqrt{\frac{R}{|u|}}.
\end{equation}
Now integrating the approximate conservation law \eqref{order1} and using \eqref{BasicEq} we obtain
\begin{eqnarray}\label{first}
&&\left| \mathcal{E}_{\mathcal{H}_{s_1}}(\slashed{\nabla}_R\underline{\phi}_a) + \mathcal{E}_{\mathcal{S}^{s_1,s_2}_{u_0}}(\slashed{\nabla}_R\underline{\phi}_a) - \mathcal{E}_{\mathcal{H}_{s_2}}(\slashed{\nabla}_R\underline{\phi}_a) \right|\cr
&\lesssim& \int_{s_1}^{s_2}\int_{\mathcal{H}_s}\left| \left( 2(1-3M)R\slashed{\nabla}_R^2\underline{\phi}_a - 2(1-6MR)\slashed{\nabla}_R\underline{\phi}_a \right)\right.\cr
&&\hspace{4cm}\times\left.\left( u^2\slashed{\nabla}_u\slashed{\nabla}_R\underline{\phi}^a - 2(1+uR)\slashed{\nabla}_R^2\underline{\phi}^a \right)\right|  \frac{1}{|u|}\d u \d^2\omega \d s \cr
&&+ \int_{s_1}^{s_2}\int_{\mathcal{H}_s} 4MR^2(3+uR)|\slashed{\nabla}^2_R\underline{\phi}_a|^2 \frac{1}{|u|}\d u \d^2\omega \d s \cr
&\lesssim& \int_{s_1}^{s_2}\int_{\mathcal{H}_s} \left(2|1-3M|s\slashed{\nabla}_R^2\underline{\phi}_a\right) \left(|u|\slashed{\nabla}_u\slashed{\nabla}_R\underline{\phi}^a \right) \frac{1}{\sqrt s}\sqrt{\frac{R}{|u|}} \d u \d^2\omega \d s \cr
&&+ \int_{s_1}^{s_2}\int_{\mathcal{H}_s} 4|(1-3M)(1+uR)||\slashed{\nabla}_R^2\underline{\phi}_a|^2 \frac{1}{\sqrt s}\frac{R}{|u|} \d u \d^2\omega \d s \cr
&&+ \int_{s_1}^{s_2}\int_{\mathcal{H}_s} 2|1-6MR|\left(\slashed{\nabla}_R\underline{\phi}_a\right) \left(|u| \slashed{\nabla}_u\slashed{\nabla}_R\underline{\phi}^a \right)\frac{1}{\sqrt s} \d u \d^2\omega \d s \cr
&&+ \int_{s_1}^{s_2}\int_{\mathcal{H}_s} 4|(1-6MR)(1+uR)|\left(\slashed{\nabla}_R\underline{\phi}_a\right) \left(\slashed{\nabla}^2_R\underline{\phi}^a \right)\frac{1}{\sqrt s}\sqrt{\frac{R}{|u|}} \d u \d^2\omega \d s \cr
&&+\int_{s_1}^{s_2}\int_{\mathcal{H}_s} \frac{1}{\sqrt s}\frac{R}{|u|}|\slashed{\nabla}^2_R\underline{\phi}_a|^2\d u \d^2\omega \d s\cr
&\lesssim& \int_{s_1}^{s_2}\frac{1}{\sqrt{s}}\int_{\mathcal{H}_s}\left( \frac{R}{|u|}|\slashed{\nabla}_R^2\underline{\phi}_a|^2 + |u|^2|\slashed{\nabla}_u\slashed{\nabla}_R\underline{\phi}_a|^2\right) \d u \d^2\omega \d s \cr
&&+ \int_{s_1}^{s_2}\frac{1}{\sqrt s}\int_{\mathcal{H}_s} \frac{R}{|u|}|\slashed{\nabla}_R^2\underline{\phi}_a|^2 \d u \d^2\omega \d s \cr
&&+ \int_{s_1}^{s_2} \frac{1}{\sqrt s}\int_{\mathcal{H}_s}\left(|\slashed{\nabla}_R\underline{\phi}_a|^2 + |u|^2 \slashed{\nabla}_u\slashed{\nabla}_R\underline{\phi}_a|^2 \right) \d u \d^2\omega \d s \cr
&&+ \int_{s_1}^{s_2} \frac{1}{\sqrt s} \int_{\mathcal{H}_s} \left(|\slashed{\nabla}_R\underline{\phi}_a|^2 + \frac{R}{|u|}|\slashed{\nabla}^2_R\underline{\phi}_a|^2 \right) \d u \d^2\omega \d s \cr
&&+\int_{s_1}^{s_2}  \frac{1}{\sqrt s}\int_{\mathcal{H}_s}\frac{R}{|u|}|\slashed{\nabla}^2_R\underline{\phi}_a|^2\d u \d^2\omega \d s\cr
&\lesssim& \int_{s_1}^{s_2}\frac{1}{\sqrt s}\left(\mathcal{E}_{\mathcal{H}_s}(\underline{\phi}_a) +  \mathcal{E}_{\mathcal{H}_s}(\slashed{\nabla}_R\underline{\phi}_a) \right)\d s.
\end{eqnarray}
Note that we can remove the factor $\dfrac{1}{\sqrt s}$ by changing variable $\tau=-2(\sqrt{s}-1)$, hence $\dfrac{\d\tau}{\d s}=-\dfrac{1}{\sqrt s}$ and $\tau$ varies from $2$ to $0$, when $s$ varies from $0$ to $1$. Therefore, we have
\begin{eqnarray*}
&&\left| \mathcal{E}_{\mathscr{H}_{\tau(s_1)}}(\slashed{\nabla}_R\underline{\phi}_a) + \mathcal{E}_{\mathcal{S}^{s_1,s_2}_{u_0}}(\slashed{\nabla}_R\underline{\phi}_a) - \mathcal{E}_{\mathscr{H}_{\tau(s_2)}}(\slashed{\nabla}_R\underline{\phi}_a) \right|\cr
&\lesssim& \int_{\tau(s_1)}^{\tau(s_2)}\left( \mathcal{E}_{\mathscr{H}_{\tau(s)}}(\slashed{\nabla}_R\underline{\phi}_a) + \mathcal{E}_{\mathscr{H}_{\tau(s)}}(\underline{\phi}_a) \right)\d s,
\end{eqnarray*}
where $\mathscr{H}_{\tau(s)}=\mathcal{H}_s$ (see also for scalar wave equations \cite[Appendix A4 and A6, pages 24-25]{MaNi2009}).

Combining inequalities \eqref{first} and \eqref{zero}, we can establish that
\begin{eqnarray}\label{Gron}
&&\left| \left(\mathcal{E}_{\mathcal{H}_{s_1}}(\slashed{\nabla}_R\underline{\phi}_a)+ \mathcal{E}_{\mathcal{H}_{s_1}}(\underline{\phi}_a)\right) + \left( \mathcal{E}_{\mathcal{S}^{s_1,s_2}_{u_0}}(\slashed{\nabla}_R\underline{\phi}_a) + \mathcal{E}_{\mathcal{S}^{s_1,s_2}_{u_0}}(\underline{\phi}_a)\right) - \left( \mathcal{E}_{\mathcal{H}_{s_2}}(\slashed{\nabla}_R\underline{\phi}_a) + \mathcal{E}_{\mathcal{H}_{s_2}}(\underline{\phi}_a) \right) \right| \cr
&\lesssim& \int_{s_1}^{s_2}\frac{1}{\sqrt s}\left( \mathcal{E}_{\mathcal{H}_s}(\slashed{\nabla}_R\underline{\phi}_a)  + \mathcal{E}_{\mathcal{H}_s}(\underline{\phi}_a) \right)\d s.
\end{eqnarray}
This inequality leads to
\begin{eqnarray}
&&\mathcal{E}_{\mathcal{H}_{s_1}}(\slashed{\nabla}_R\underline{\phi}_a) + \mathcal{E}_{\mathcal{H}_{s_1}}(\underline{\phi}_a) \lesssim \mathcal{E}_{\mathcal{H}_1}(\slashed{\nabla}_R\underline{\phi}_a) + \mathcal{E}_{\mathcal{H}_1}(\underline{\phi}_a) \cr
&&\hspace{5cm}+ \int_{s_1}^1 \frac{1}{\sqrt s}\left( \mathcal{E}_{\mathcal{H}_s}(\slashed{\nabla}_R\underline{\phi}_a) + \mathcal{E}_{\mathcal{H}_s}(\underline{\phi}_a) \right) \d s, \label{00}\\
&&\mathcal{E}_{\mathcal{H}_{s_2}}(\slashed{\nabla}_R\underline{\phi}_a) + \mathcal{E}_{\mathcal{H}_{s_2}}(\underline{\phi}_a) \lesssim \mathcal{E}_{\mathcal{S}^+_{u_0}}(\underline{\phi}_a) + \mathcal{E}_{\scri^+_{u_0}}(\underline{\phi}_a) \cr
&&\hspace{5cm}+ \int_0^{s_2}\frac{1}{\sqrt s}\left( \mathcal{E}_{\mathcal{H}_s}(\slashed{\nabla}_R\underline{\phi}_a) + \mathcal{E}_{\mathcal{H}_s}(\underline{\phi}_a) \right)\d s \label{11}
\end{eqnarray}
for all $0\leq s_1<s_2\leq 1$. Since the function $\dfrac{1}{\sqrt{s}}$ is integrable on $[0,\,1]$, we use Gronwall’s inequality for inequalities \eqref{00} and \eqref{11} with the scalar function $z(s)=\mathcal{E}_{\mathcal{H}_s}(\slashed{\nabla}_R\underline{\phi}_a) + \mathcal{E}_{\mathcal{H}_s}(\underline{\phi}_a)$, and get the following result of energy estimates at first order.
\begin{theorem}
For $u_0<0$ and $|u_0|$ large enough  and for any smooth compactly supported initial data
at $\Sigma_0= \left\{ t = 0 \right\}$, the associated solutions $\underline{\phi}_a$ of \eqref{ReFac02}
satisfying that
$$\mathcal{E}_{\scri^+_{u_0}}(\slashed{\nabla}_R\underline{\phi}_a) + \mathcal{E}_{\scri^+_{u_0}}(\underline{\phi}_a) \lesssim \mathcal{E}_{\mathcal{H}_1}(\slashed{\nabla}_R\underline{\phi}_a) + \mathcal{E}_{\mathcal{H}_1}(\underline{\phi}_a),$$
$$\mathcal{E}_{\mathcal{H}_1}(\slashed{\nabla}_R\underline{\phi}_a) + \mathcal{E}_{\mathcal{H}_1}(\underline{\phi}_a) \lesssim \mathcal{E}_{\scri^+_{u_0}}(\slashed{\nabla}_R\underline{\phi}_a) + \mathcal{E}_{\scri^+_{u_0}}(\underline{\phi}_a) + + \mathcal{E}_{{\mathcal{S}^+_{u_0}}}(\slashed{\nabla}_R\underline{\phi}_a) + \mathcal{E}_{{\mathcal{S}^+_{u_0}}}(\underline{\phi}_a).$$
\end{theorem}
By the same way we have the higher order estimates for $\slashed{\nabla}_R\underline{\phi}_a$ in the following theorem.
\begin{theorem}\label{P2}
For $u_0<0$ and $|u_0|$ large enough  and for any smooth compactly supported initial data
at $\Sigma_0= \left\{ t = 0 \right\}$, the associated solutions $\underline{\phi}_a$ of \eqref{ReFac02}
satisfying that
$$\sum_{p=0}^k\mathcal{E}_{\scri^+_{u_0}}(\slashed{\nabla}_R^p\underline{\phi}_a) \lesssim \sum_{p=0}^k\mathcal{E}_{\mathcal{H}_1}(\slashed{\nabla}_R^p\underline{\phi}_a),$$
$$\sum_{p=0}^k\mathcal{E}_{\mathcal{H}_1}(\slashed{\nabla}_R^p\underline{\phi}_a) \lesssim \sum_{p=0}^k \left( \mathcal{E}_{\scri^+_{u_0}}(\slashed{\nabla}^p_R\underline{\phi}_a) + \mathcal{E}_{{\mathcal{S}^+_{u_0}}}(\slashed{\nabla}^p_R\underline{\phi}_a) \right)$$
for all $k\in \mathbb{N}$.
\end{theorem}
Combining the two theorems \ref{P1} and \ref{P2} we obtain the two-side estimates of the energies for all covariant derivatives of tensorial fields.
\begin{theorem}\label{PeelingFac}
For $u_0<0$ and $|u_0|$ large enough  and for any smooth compactly supported initial data
at $\Sigma_0= \left\{ t = 0 \right\}$, the associated solutions $\underline{\phi}_a$ of \eqref{ReFac02}
satisfying that
$$\sum_{m+n+p\leq k}\mathcal{E}_{\scri^+_{u_0}}(\slashed{\nabla}_u^m\slashed{\nabla}_R^n\slashed{\nabla}_{\mathbb{S}^2}^p\underline{\phi}_a) \lesssim \sum_{m+n+p\leq k}\mathcal{E}_{\mathcal{H}_1}(\slashed{\nabla}_u^m\slashed{\nabla}_R^n\slashed{\nabla}_{\mathbb{S}^2}^p\underline{\phi}_a),$$
$$\sum_{m+n+p\leq k}\mathcal{E}_{\mathcal{H}_1}(\slashed{\nabla}_u^m\slashed{\nabla}_R^n\slashed{\nabla}_{\mathbb{S}^2}^p\underline{\phi}_a) \lesssim \sum_{m+n+p\leq k}\mathcal{E}_{\scri^+_{u_0}}(\slashed{\nabla}_u^m\slashed{\nabla}_R^n\slashed{\nabla}_{\mathbb{S}^2}^p\underline{\phi}_a) + \mathcal{E}_{{\mathcal{S}^+_{u_0}}}(\slashed{\nabla}_u^m\slashed{\nabla}_R^n\slashed{\nabla}_{\mathbb{S}^2}^p\underline{\phi}_a),$$
where $k,m,n,p\in \mathbb{N}$.
\end{theorem}
Now we give the definition of the peeling at order $k \in \mathbb{N}$:
\begin{definition}
A solution $\underline{\phi}_a$ of the tensorial Fackerell equation \eqref{ReFac02} in $\Omega_{u_0}^+$
peels at order $k\in \mathbb{N}$ if $\underline{\phi}_a$ satisfies
$$\sum_{m+n+p\leq k}\mathcal{E}_{\scri^+_{u_0}}(\slashed{\nabla}_u^m\slashed{\nabla}_R^n\slashed{\nabla}_{\mathbb{S}^2}^p\underline{\phi}_a)<+\infty.$$
\end{definition}
Theorem \ref{PeelingFac} gives us a characterization of the class of initial data on $\mathcal{H}_1$ that guarantees that the corresponding solution peels at a given order $k\in \mathbb{N}$; it is the completion of smooth compactly supported data on $\mathcal{H}_1$ in the norm
$$\left( \sum_{m+n+p\leq k}\mathcal{E}_{\mathcal{H}_1}(\slashed{\nabla}_u^m\slashed{\nabla}_R^n\slashed{\nabla}_{\mathbb{S}^2}^p\underline{\phi}_a) \right)^{1/2}.$$

\subsection{Peeling for Teukolsky equations}\label{PeeTeu}
Integrating the approximate conservation law \eqref{zero2} on the domain $\Omega_{u_0}^{s_1,s_2}$ given by \eqref{Omega} we obtain
\begin{eqnarray}\label{order00}
&&\left| \mathcal{E}_{\mathcal{H}_{s_1}}(\widehat{\underline{\alpha}}_a) + \mathcal{E}_{\mathcal{S}^{s_1,s_2}_{u_0}}(\widehat{\underline{\alpha}}_a) - \mathcal{E}_{\mathcal{H}_{s_2}}(\widehat{\underline{\alpha}}_a) \right|\cr
&\simeq& \left|  \int_{s_1}^{s_2}\int_{\mathcal{H}_s} -2R(1-3MR)\slashed{\nabla}_R\widehat{\underline\alpha}_a \left( u^2\slashed{\nabla}_u \widehat{\underline{\alpha}}^a -2(1+uR)\slashed{\nabla}_R\widehat{\underline{\alpha}}^a \right) \frac{1}{|u|}\d u \d^2\omega \d s\right|\cr
&&+\left| \int_{s_1}^{s_2}\int_{\mathcal{H}_s} 4MR^2(3+uR)|\slashed{\nabla}_R\widehat{\underline{\alpha}}_a|^2 \frac{1}{|u|}\d u \d^2\omega \d s \right|\cr
&\lesssim& \int_{s_1}^{s_2}\int_{\mathcal{H}_s} 2R|1-3MR| \left(\slashed{\nabla}_R\widehat{\underline\alpha}_a\right) \left(|u|\slashed{\nabla}_u \widehat{\underline{\alpha}}^a \right) \d u \d^2\omega \d s\cr
&&+ \int_{s_1}^{s_2}\int_{\mathcal{H}_s} 4|(1-3MR)(1+uR)| |\slashed{\nabla}_R\widehat{\underline\alpha}_a|^2 \frac{R}{|u|} \d u \d^2\omega \d s\cr
&&+\int_{s_1}^{s_2}\int_{\mathcal{H}_s} 4MR^2(3+uR)|\slashed{\nabla}_R\widehat{\underline{\alpha}}_a|^2 \frac{1}{|u|}\d u \d^2\omega \d s\cr
&\lesssim& \int_{s_1}^{s_2}\frac{1}{\sqrt s}\int_{\mathcal{H}_s} \left( R^2|\slashed{\nabla}_R\widehat{\underline\alpha}_a|^2 +  |u|^2|\slashed{\nabla}_u \widehat{\underline{\alpha}}^a|^2 \right) \d u \d^2\omega \d s\cr
&&+ \int_{s_1}^{s_2}\frac{1}{\sqrt s}\int_{\mathcal{H}_s} \frac{R}{|u|} |\slashed{\nabla}_R\widehat{\underline\alpha}_a|^2 \d u \d^2\omega \d s\cr
&&+\int_{s_1}^{s_2}\frac{1}{\sqrt s} \int_{\mathcal{H}_s} \frac{R}{|u|}|\slashed{\nabla}_R\widehat{\underline{\alpha}}_a|^2\d u \d^2\omega \d s\cr
&\lesssim& \int_{s_1}^{s_2}\frac{1}{\sqrt s}\int_{\mathcal{H}_s} \left( \frac{sR}{|u|}|\slashed{\nabla}_R\widehat{\underline\alpha}_a|^2 +  |u|^2|\slashed{\nabla}_u \widehat{\underline{\alpha}}^a|^2 \right) \d u \d^2\omega \d s\cr
&&+\int_{s_1}^{s_2}\frac{1}{\sqrt s}\mathcal{E}_{\mathcal{H}_s}(\widehat{\underline{\alpha}}_a) \d s\cr
&\lesssim& \int_{s_1}^{s_2}\frac{1}{\sqrt s}\mathcal{E}_{\mathcal{H}_s}(\widehat{\underline{\alpha}}_a) \d s.
\end{eqnarray}
Similar to the energy estimate for the Fackerell-Ipser equation at zero order, the factor $\dfrac{1}{\sqrt{s}}$ is presented to serve the higher order estimates.

Since the function $\dfrac{1}{\sqrt{s}}$ is integrable on $[0,\,1]$, Gronwall’s inequality entails the following result for energy estimates at zero order (by the same way in Subsection \ref{Sub1}).
\begin{theorem}\label{P011}
For $u_0<0$ and $|u_0|$ large enough  and for any smooth compactly supported initial data
at $\Sigma_0= \left\{ t = 0 \right\}$, the associated solutions $\underline{\phi}_a$ of \eqref{ReFac02}
satisfying that
$$\mathcal{E}_{\scri^+_{u_0}}(\widehat{\underline{\alpha}}_a) \lesssim \mathcal{E}_{\mathcal{H}_1}(\widehat{\underline{\alpha}}_a),$$
$$\mathcal{E}_{\mathcal{H}_1}(\widehat{\underline{\alpha}}_a) \lesssim \mathcal{E}_{\scri^+_{u_0}}(\widehat{\underline{\alpha}}_a) + \mathcal{E}_{\mathcal{H}_{\mathcal{S}^+_{u_0}}}(\widehat{\underline{\alpha}}_a).$$
\end{theorem}

Since the approximate conservation law \eqref{zero2} is valid for $\slashed{\nabla}_u^k\widehat{\underline{\alpha}}_a$ and $\slashed{\nabla}_{\mathbb{S}^2}^k\widehat{\underline{\alpha}}_a$ with $k\in \mathbb{N}$, we have the following theorem.
\begin{theorem}\label{P11}
For $u_0<0$ and $|u_0|$ large enough  and for any smooth compactly supported initial data
at $\Sigma_0= \left\{ t = 0 \right\}$, the associated solutions $\underline{\phi}_a$ of \eqref{ReFac02}
satisfying that
$$\mathcal{E}_{\scri^+_{u_0}}(\slashed{\nabla}_u^k\slashed{\nabla}_{\mathbb{S}^2}^l\widehat{\underline{\alpha}}_a) \lesssim \mathcal{E}_{\mathcal{H}_1}(\slashed{\nabla}_u^k\slashed{\nabla}_{\mathbb{S}^2}^l\widehat{\underline{\alpha}}_a),$$
$$\mathcal{E}_{\mathcal{H}_1}(\slashed{\nabla}_u^k\slashed{\nabla}_{\mathbb{S}^2}^l\widehat{\underline{\alpha}}_a) \lesssim \mathcal{E}_{\scri^+_{u_0}}(\slashed{\nabla}_u^k\slashed{\nabla}_{\mathbb{S}^2}^l\widehat{\underline{\alpha}}_a) + \mathcal{E}_{\mathcal{H}_{\mathcal{S}^+_{u_0}}}(\slashed{\nabla}_u^k\slashed{\nabla}_{\mathbb{S}^2}^l\widehat{\underline{\alpha}}_a),$$
where $k,l \in \mathbb{N}$.
\end{theorem}

Now integrating the approximate conservation law \eqref{order1} and using \eqref{BasicEq} we obtain
\begin{eqnarray}\label{order11}
&&\left| \mathcal{E}_{\mathcal{H}_{s_1}}(\slashed{\nabla}_R\widehat{\underline{\alpha}}_a) + \mathcal{E}_{\mathcal{S}^{s_1,s_2}_{u_0}}(\slashed{\nabla}_R\widehat{\underline{\alpha}}_a) - \mathcal{E}_{\mathcal{H}_{s_2}}(\slashed{\nabla}_R\widehat{\underline{\alpha}}_a) \right|\cr
&\lesssim& \int_{s_1}^{s_2}\int_{\mathcal{H}_s}\left| \left( 6MR(R-1)\slashed{\nabla}_R^2\widehat{\underline{\alpha}}_a - 4(1-6MR)\slashed{\nabla}_R\widehat{\underline{\alpha}}_a \right)\right.\cr
&&\hspace{4cm}\times\left.\left( u^2\slashed{\nabla}_u\slashed{\nabla}_R\widehat{\underline{\alpha}}^a - 2(1+uR)\slashed{\nabla}_R^2\widehat{\underline{\alpha}}^a \right)\right|  \frac{1}{|u|}\d u \d^2\omega \d s \cr
&&+ \int_{s_1}^{s_2}\int_{\mathcal{H}_s} 4MR^2(3+uR)|\slashed{\nabla}^2_R\widehat{\underline{\alpha}}_a|^2 \frac{1}{|u|}\d u \d^2\omega \d s \cr
&\lesssim& \int_{s_1}^{s_2}\int_{\mathcal{H}_s} \left(6M|R-1|s\slashed{\nabla}_R^2\widehat{\underline{\alpha}}_a\right) \left(|u|\slashed{\nabla}_u\slashed{\nabla}_R\widehat{\underline{\alpha}}^a \right) \frac{1}{\sqrt s}\sqrt{\frac{R}{|u|}} \d u \d^2\omega \d s \cr
&&+ \int_{s_1}^{s_2}\int_{\mathcal{H}_s} 12|M(R-1)(1+uR)||\slashed{\nabla}_R^2\widehat{\underline{\alpha}}_a|^2 \frac{1}{\sqrt s}\frac{R}{|u|} \d u \d^2\omega \d s \cr
&&+ \int_{s_1}^{s_2}\int_{\mathcal{H}_s} 4|1-6MR|\left(\slashed{\nabla}_R\widehat{\underline{\alpha}}_a\right) \left(|u| \slashed{\nabla}_u\slashed{\nabla}_R\widehat{\underline{\alpha}}^a \right)\frac{1}{\sqrt s} \d u \d^2\omega \d s \cr
&&+ \int_{s_1}^{s_2}\int_{\mathcal{H}_s} 4|(1-6MR)(1+uR)|\left(\slashed{\nabla}_R\widehat{\underline{\alpha}}_a\right) \left(\slashed{\nabla}^2_R\widehat{\underline{\alpha}}^a \right)\frac{1}{\sqrt s}\sqrt{\frac{R}{|u|}} \d u \d^2\omega \d s \cr
&&+\int_{s_1}^{s_2}\int_{\mathcal{H}_s} \frac{1}{\sqrt s}\frac{R}{|u|}|\slashed{\nabla}^2_R\widehat{\underline{\alpha}}_a|^2\d u \d^2\omega \d s\cr
&\lesssim& \int_{s_1}^{s_2}\int_{\mathcal{H}_s}\frac{1}{\sqrt{s}}\left( \frac{R}{|u|}|\slashed{\nabla}_R^2\widehat{\underline{\alpha}}_a|^2 + |u|^2|\slashed{\nabla}_u\slashed{\nabla}_R\widehat{\underline{\alpha}}_a|^2\right) \d u \d^2\omega \d s \cr
&&+ \int_{s_1}^{s_2}\int_{\mathcal{H}_s} \frac{1}{\sqrt s}\frac{R}{|u|}|\slashed{\nabla}_R^2\widehat{\underline{\alpha}}_a|^2 \d u \d^2\omega \d s \cr
&&+ \int_{s_1}^{s_2}\int_{\mathcal{H}_s} \frac{1}{\sqrt s}\left(|\slashed{\nabla}_R\widehat{\underline{\alpha}}_a|^2 + |u|^2 \slashed{\nabla}_u\slashed{\nabla}_R\widehat{\underline{\alpha}}_a|^2 \right) \d u \d^2\omega \d s \cr
&&+ \int_{s_1}^{s_2}\int_{\mathcal{H}_s} \frac{1}{\sqrt s} \left(|\slashed{\nabla}_R\widehat{\underline{\alpha}}_a|^2 + \frac{R}{|u|}|\slashed{\nabla}^2_R\widehat{\underline{\alpha}}_a|^2 \right) \d u \d^2\omega \d s \cr
&&+\int_{s_1}^{s_2}\int_{\mathcal{H}_s} \frac{1}{\sqrt s}\frac{R}{|u|}|\slashed{\nabla}^2_R\widehat{\underline{\alpha}}_a|^2\d u \d^2\omega \d s\cr
&\lesssim& \int_{s_1}^{s_2}\frac{1}{\sqrt s}\left(\mathcal{E}_{\mathcal{H}_s}(\widehat{\underline{\alpha}}_a) +  \mathcal{E}_{\mathcal{H}_s}(\slashed{\nabla}_R\widehat{\underline{\alpha}}_a) \right)\d s.
\end{eqnarray}
Combining inequalities \eqref{order00} and \eqref{order11}, we can establish that
\begin{eqnarray}\label{Gron'}
&&\left| \left(\mathcal{E}_{\mathcal{H}_{s_1}}(\slashed{\nabla}_R\widehat{\underline{\alpha}}_a)+ \mathcal{E}_{\mathcal{H}_{s_1}}(\widehat{\underline{\alpha}}_a)\right) + \left( \mathcal{E}_{\mathcal{S}^{s_1,s_2}_{u_0}}(\slashed{\nabla}_R\widehat{\underline{\alpha}}_a) + \mathcal{E}_{\mathcal{S}^{s_1,s_2}_{u_0}}(\widehat{\underline{\alpha}}_a)\right) - \left( \mathcal{E}_{\mathcal{H}_{s_2}}(\slashed{\nabla}_R\widehat{\underline{\alpha}}_a) + \mathcal{E}_{\mathcal{H}_{s_2}}(\widehat{\underline{\alpha}}_a) \right) \right| \cr
&\lesssim& \int_{s_1}^{s_2}\frac{1}{\sqrt s}\left( \mathcal{E}_{\mathcal{H}_s}(\slashed{\nabla}_R\widehat{\underline{\alpha}}_a) + \mathcal{E}_{\mathcal{H}_s}(\widehat{\underline{\alpha}}_a) \right)\d s.
\end{eqnarray}
Since the function $\dfrac{1}{\sqrt{s}}$ is integrable on $[0,\,1]$, Gronwall’s inequality entails the following result (by the same way in Subsection \ref{Sub1}).
\begin{theorem}
For $u_0<0$ and $|u_0|$ large enough  and for any smooth compactly supported initial data
at $\Sigma_0= \left\{ t = 0 \right\}$, the associated solutions $\underline{\phi}_a$ of \eqref{ReTeu2}
satisfying that
$$\mathcal{E}_{\scri^+_{u_0}}(\slashed{\nabla}_R\widehat{\underline{\alpha}}_a) + \mathcal{E}_{\scri^+_{u_0}}(\widehat{\underline{\alpha}}_a) \lesssim \mathcal{E}_{\mathcal{H}_1}(\slashed{\nabla}_R\widehat{\underline{\alpha}}_a) + \mathcal{E}_{\mathcal{H}_1}(\widehat{\underline{\alpha}}_a),$$
$$\mathcal{E}_{\mathcal{H}_1}(\slashed{\nabla}_R\widehat{\underline{\alpha}}_a) + \mathcal{E}_{\mathcal{H}_1}(\widehat{\underline{\alpha}}_a) \lesssim \mathcal{E}_{\scri^+_{u_0}}(\slashed{\nabla}_R\widehat{\underline{\alpha}}_a) + \mathcal{E}_{\scri^+_{u_0}}(\widehat{\underline{\alpha}}_a) + \mathcal{E}_{{\mathcal{S}^+_{u_0}}}(\slashed{\nabla}_R\widehat{\underline{\alpha}}_a) + \mathcal{E}_{{\mathcal{S}^+_{u_0}}}(\widehat{\underline{\alpha}}_a).$$
\end{theorem}
By the same way, we have the higher order estimates for $\slashed{\nabla}_R\widehat{\underline{\alpha}}_a$ in the following theorem.
\begin{theorem}\label{P22}
For $u_0<0$ and $|u_0|$ large enough  and for any smooth compactly supported initial data
at $\Sigma_0= \left\{ t = 0 \right\}$, the associated solutions $\underline{\phi}_a$ of \eqref{ReTeu2}
satisfying that
$$\mathcal{E}_{\scri^+_{u_0}}(\slashed{\nabla}_R^k\widehat{\underline{\alpha}}_a) \lesssim \sum_{p=0}^k\mathcal{E}_{\mathcal{H}_1}(\slashed{\nabla}_R^p\widehat{\underline{\alpha}}_a),$$
$$\mathcal{E}_{\mathcal{H}_1}(\slashed{\nabla}_R^k\widehat{\underline{\alpha}}_a) \lesssim \sum_{p=0}^k \left( \mathcal{E}_{\scri^+_{u_0}}(\slashed{\nabla}^p_R\widehat{\underline{\alpha}}_a) + \mathcal{E}_{{\mathcal{S}^+_{u_0}}}(\slashed{\nabla}^p_R\widehat{\underline{\alpha}}_a) \right)$$
for all $k,m,n,p\in \mathbb{N}$.
\end{theorem}
Combining the two theorems \ref{P11} and \ref{P22} we obtain the two-side estimates of the energies for all covariant derivatives of tensorial fields.
\begin{theorem}\label{PeelingTeu}
For $u_0<0$ and $|u_0|$ large enough  and for any smooth compactly supported initial data
at $\Sigma_0= \left\{ t = 0 \right\}$, the associated solutions $\underline{\phi}_a$ of \eqref{ReTeu2}
satisfying that
$$\sum_{m+n+p\leq k}\mathcal{E}_{\scri^+_{u_0}}(\slashed{\nabla}_u^m\slashed{\nabla}_R^n\slashed{\nabla}_{\mathbb{S}^2}^p\widehat{\underline{\alpha}}_a) \lesssim \sum_{m+n+p\leq k}\mathcal{E}_{\mathcal{H}_1}(\slashed{\nabla}_u^m\slashed{\nabla}_R^n\slashed{\nabla}_{\mathbb{S}^2}^p\widehat{\underline{\alpha}}_a),$$
$$\sum_{m+n+p\leq k}\mathcal{E}_{\mathcal{H}_1}(\slashed{\nabla}_u^m\slashed{\nabla}_R^n\slashed{\nabla}_{\mathbb{S}^2}^p\widehat{\underline{\alpha}}_a) \lesssim \sum_{m+n+p\leq k}\mathcal{E}_{\scri^+_{u_0}}(\slashed{\nabla}_u^m\slashed{\nabla}_R^n\slashed{\nabla}_{\mathbb{S}^2}^p\widehat{\underline{\alpha}}_a) + \mathcal{E}_{{\mathcal{S}^+_{u_0}}}(\slashed{\nabla}_u^m\slashed{\nabla}_R^n\slashed{\nabla}_{\mathbb{S}^2}^p\widehat{\underline{\alpha}}_a),$$
where $k,m,n,p\in \mathbb{N}$.
\end{theorem}
Now we give the definition of the peeling at order $k \in \mathbb{N}$:
\begin{definition}
A solution $\underline{\phi}_a$ of the tensorial Fackerell equation \eqref{ReTeu2} in $\Omega_{u_0}^+$
peels at order $k\in \mathbb{N}$ if $\widehat{\underline{\alpha}}_a$ satisfies
$$\sum_{m+n+p\leq k}\mathcal{E}_{\scri^+_{u_0}}(\slashed{\nabla}_u^m\slashed{\nabla}_R^n\slashed{\nabla}_{\mathbb{S}^2}^p\widehat{\underline{\alpha}}_a)<+\infty.$$
\end{definition}
Theorem \ref{PeelingTeu} gives us a characterization of the class of initial data on $\mathcal{H}_1$ that guarantees that the corresponding solution peels at a given order $k\in \mathbb{N}$; it is the completion of smooth compactly supported data on $\mathcal{H}_1$ in the norm
$$\left( \sum_{m+n+p\leq k}\mathcal{E}_{\mathcal{H}_1}(\slashed{\nabla}_u^m\slashed{\nabla}_R^n\slashed{\nabla}_{\mathbb{S}^2}^p\widehat{\underline{\alpha}}_a) \right)^{1/2}.$$
\begin{remark}
\noindent
\item[$\bullet$] By the same way we can establish the peeling for the tensorial Fackerell-Ipser and spin $+1$ equations \eqref{ReFac01} and \eqref{ReTeu1} respectively. This means that we establish the asymptotic behaviours of associated solutions of \eqref{ReFac01} and \eqref{ReTeu1} along the incoming radial geodesics in coordinates $(v,\, R,\, \theta,\, \varphi)$.
\item[$\bullet$] It seems that the results in this paper can be extended to the tensorial Regge-Wheeler and spin $\pm 2$ Teukolsky equations on spherically symmetric black hole spacetimes such as Schwarzschild and Reissner-Nordstr\"om de Sitter spacetimes, where some recent works \cite{Da2019, El2020', Masao} can be useful.
\item[$\bullet$] The extension of peeling for tensorial Fackerell-Ipser, tensorial Regge-Wheeler, spin $\pm 1$ and spin $\pm 2$ Teukolsky equations on Kerr spacetime is an interesting question, where our work \cite{NiXuan2019} can be useful. We hope to treat this problem in a forthcoming paper.
\end{remark}

\end{document}